\documentclass[11pt]{article}
\usepackage{amsmath,amsthm,amsfonts,graphicx,color,caption,subcaption,float,fullpage}
\usepackage{authblk}
 
\usepackage{microtype,amssymb}
\newtheorem{thm}{Theorem}
\newtheorem{lem}[thm]{Lemma}
\newtheorem{cor}[thm]{Corollary}
\newtheorem{dfn}[thm]{Definition}
\newtheorem{prp}[thm]{Proposition}

\newtheorem{ass}[thm]{Assumption}

\newtheorem{rmk}[thm]{Remark}
\def\RR{{\mathbb R}}
\def\EE{{\mathbb E}}

\begin{document}

\title{Randomized embeddings with slack, and high-dimensional Approximate Nearest Neighbor}
\author[1]{Evangelos Anagnostopoulos}
\author[1]{Ioannis Z. Emiris}
\author[1]{Ioannis Psarros}
\affil[1]{Dept.\ of Informatics \& Telecommunications, University of Athens, Greece}

\maketitle

\begin{abstract}
The approximate nearest neighbor problem ($\epsilon$-ANN) in high dimensional 
Euclidean space has been mainly addressed
by Locality Sensitive Hashing (LSH), which has polynomial dependence
in the dimension, sublinear query time, 
but subquadratic space requirement.
In this paper, we 
introduce a new definition of 
``low-quality'' embeddings for metric spaces.  
It requires that, for some query point $q$, there exists an approximate nearest 
neighbor among the pre-images of the $k>1$ approximate 
nearest neighbors in the target space. Focusing on 
Euclidean spaces, we employ random projections in order to 
reduce the original problem to one in a space of dimension inversely proportional to $k$. 

The $k$ approximate nearest 
neighbors can be efficiently retrieved 
by a data structure such as BBD-trees. 
The same approach is applied 
to the problem of computing an approximate near neighbor, where we obtain a 
data structure requiring 
linear space, 
and query time in $O(d n^{\rho})$,
for $\rho\approx 1-\epsilon^2/\log(1/\epsilon)$. 
This directly implies a solution for $\epsilon$-ANN, while achieving a better exponent in the query time than the method based on BBD-trees. 
Better bounds are obtained in the case of 
doubling subsets of $\ell_2$, 
by combining our method with $r$-nets.

We implement our method in C++, and present experimental results in dimension
up to $500$ and $10^6$ points, which 
show that performance is better 
than predicted by the analysis. In addition, we compare our ANN approach 
to E2LSH, which implements LSH, and we show that the theoretical advantages of each method are reflected on their actual performance.
\end{abstract}



\section{Introduction}

Nearest neighbor searching is a fundamental computational problem with several applications in Computer Science and beyond. Let us focus on the Euclidean version of the problem. 
Let $X$ be a set of $n$ points in $d$-dimensional Euclidean space $\ell_2^d$.
We denote by $\|\cdot \|$ the inherent Euclidean norm $\| \cdot\|_2$. The problem consists in building a data structure such that for any query point $q$, one may report 
a point $p \in X$ for which $\|p- q\| \leq \|p' - q\|$, for all $p' \in X$; then $p$ is said to be a “nearest neighbor” of $q$.
However, an exact solution to high-dimensional 
nearest neighbor search in sublinear time 
requires prohibitively heavy resources. 
Thus,
most approaches focus on the less demanding and more relevant task of computing the approximate nearest neighbor, or $\epsilon$-ANN. 
Given a real parameter $\epsilon >0$, a 
$(1 +\epsilon )$-approximate nearest neighbor to a query point $q\in\RR^d$ is a point $p$ in $X$ such that 
$$
\|q- p\| \leq (1 + \epsilon ) \cdot \|q- p'\|, \mbox{ for all } p' \in X.
$$
Hence, under approximation, the answer can be any point whose distance from $q$ is at most $(1 +\epsilon )$ times 
larger than the distance between $q$ and its true nearest neighbor.

The corresponding augmented decision problem (with witness) is known as the  near neighbor problem, and is defined as follows. A data structure for the approximate near neighbor problem ($(\epsilon,R)$-ANN) satisfies the following conditions: if there exists 
some point $p$ in $X$ such that 
$\|q- p\| \leq  R$, then an algorithm solving this problem reports $p'\in X$ such that $\|q- p'\| \leq (1+\epsilon) \cdot R$, whereas if there is no point $p$ in $X$ such that 
$\|q- p\| \leq  (1+\epsilon)\cdot R$, then the algorithm reports ``Fail''. It is known that one can solve not-so-many instances of the decision problem with witness 
and obtain a solution for the 
$\epsilon$-ANN problem. 

\paragraph*{Our contribution}
Deterministic space partitioning techniques 
perform well in solving $\epsilon$-ANN when the dimension is relatively low, but
are affected by the curse of dimensionality. To address this issue,
randomized methods such as Locality Sensitive Hashing (LSH) are more efficient when the dimension is high. 
One might try applying the celebrated Johnson-Lindenstrauss Lemma, followed by standard space partitioning techniques, 
but the properties of the projected pointset are too strong for designing an overall efficient ANN search method (cf.\ Section~\ref{Sexisting}).

We introduce a notion of ``low-quality" randomized embeddings 
and we employ standard random projections \`{a} la Johnson-Lindenstrauss in order to define a mapping to $\ell_2^{d'}$, 
for 
$$
d'=O(\log \frac{n}{k}/\epsilon^2),
$$
such that an approximate nearest neighbor of the query lies among the pre-images of $k$ approximate nearest neighbors in the 
projected space. 
Moreover, an analogous statement can be made for 
the augmented decision problem of reporting an $(\epsilon,R)$-ANN, which also implies a solution for the $\epsilon$-ANN problem thanks to known results. 
Both of our methods employ optimal space, avoid the curse of dimensionality and lead to competitive query times.
While the first approach is more straightforward, the second outperforms the first in terms of complexity, due to a simpler auxiliary data structure.
However, reducing $\epsilon$-ANN to 
$(\epsilon,R)$-ANN is non-trivial and might not lead to fast methods in practice.

The first method leads to Theorem~\ref{Tann}, which offers a new randomized algorithm for
approximate nearest neighbor search with the following complexities.
Given $n$ points in $\ell_2^d$, the data structure, which is based on Balanced Box-Decomposition (BBD) trees,
requires optimal $O(dn)$ space,
and reports an $(1+\epsilon)^3$-approximate nearest neighbor with query 
time in $O(d n^{\rho} \log n)$, where function $\rho= 1-\Theta(\epsilon^2/\ln \ln n)$, for  $\epsilon\in(0,1/2]$, and 
shall be fully specified
in Section~\ref{Sann}.
The total preprocessing time is $O(d n \log n)$.
For each query $q \in \mathbb{R}^d$, the preprocessing phase 
succeeds with constant probability.
The low-quality embedding is extended to
finite subsets of $\ell_2$ with bounded expansion rate $c$ (see Subsection~\ref{SSexpansion} for 
definitions). 
The pointset is now mapped to a space of dimension $O(\log c)$, 
and each query costs roughly $O((c^{\log \log c})d \log n)$.

The second method applies
the same ideas to the augmented decision version of the problem.
This problem is known to be as hard as $\epsilon$-ANN (up to polylogarithmic factors). However, this simplification allows us 
to combine the aforementioned randomized embeddings with simpler data structures in the reduced space. This is the topic of 
Section~\ref{Snear}, and Theorem~\ref{thmbadnear} states that there exists a randomized data structure with linear space and linear preprocessing time which, 
for any query $q\in \RR^d$, reports an $(1+\epsilon)^3$-approximate near neighbor (or a negative answer) in time 
$O(d n^{\rho})$, where $\rho =  1-\Theta({\epsilon^2}/{\log(1/\epsilon)})$. We are able to extend our results to doubling subsets of $\ell_2$ 
(see Subsection~\ref{SDdembed} for definitions) by applying our approach to an $r$-net of the input pointset. 
The resulting data structure has linear space, preprocessing time which depends on the time required to compute an $r$-net, and query time 
$\Big({2}/{\epsilon}\Big)^{O(ddim(X))}$, where $ddim(X)$ is the doubling dimension of $X$.

We also present experiments, based on synthetic 
and image datasets, that validate our approach and our analysis. We implement our low quality embedding method in C++ and present experimental results in
up to $500$ dimensions and $10^6$ points.
One set of inputs, along with the queries, follows the ``planted nearest neighbor model'' which will 
be specified in Section~\ref{Sexperiment}. In another scenario, we assume that the near neighbors of each query point 
follow the Gaussian distribution. We also used the ANN\_SIFT1M~\cite{sift} dataset which contains a collection of $1$ million vectors in $128$ dimensions that represent images. Apart from showing that the embedding has the desired properties in practice, specifically those of Lemma~\ref{lemBBD}, 
we also implement our overall approach for computing $\epsilon$-ANN using the {\tt ANN} library for BBD-trees, and we compare with an LSH implementation, 
namely E2LSH. We show that the theoretical advantages of each method are reflected in practice.

The notation of key quantities is the same throughout the paper.

The paper extends and improves ideas from \cite{AEP15}, except for Section~\ref{Snear} which is entirely new, which
achieves better complexity bounds with a conceptually simpler data structure. 

\paragraph*{Paper organization}
The next section offers a survey of existing techniques. 
Section~\ref{Sembed} introduces our embeddings to dimension lower
than predicted by the Johnson-Linderstrauss Lemma.
Section~\ref{Sann} states our main results about $\epsilon$-ANN search in $\ell_2$ and 
for points with bounded expansion rate. 
Section~\ref{Snear} extends our ideas to the $(\epsilon,R)$-ANN problem in $\ell_2$ and in 
doubling subsets of $\ell_2$. 
Section~\ref{Sexperiment} presents experiments to validate our approach.
We conclude with open questions.

\section{Existing work}\label{Sexisting}

This section details the relevant results that existed prior to this work.

As mentioned above, an exact solution to high-dimensional nearest neighbor search, in sublinear time, 
requires heavy resources. One notable approach to the problem \cite{Mei93} shows that nearest neighbor queries 
can be answered in $O(d^5 \log n)$ time, using $O(n^{d+\delta})$ space, for arbitrary $\delta>0$.

In \cite{AMN+98}, they introduced the
Balanced Box-Decomposition (BBD) trees.
The BBD-trees data structure achieves query time $O(c \log n)$ with
$c \leq d/2 \lceil 1+6d/\epsilon \rceil^d$, using space in $O(dn)$, and
preprocessing time in $O(d n \log n)$.
BBD-trees can be used to retrieve the $k\ge 1$ approximate nearest-neighbors 
at an extra cost of $O(d\log n)$ per neighbor.
BBD-trees have proved to be very practical, as well, and have been
implemented in software library {\tt ANN}.

Another relevant data structure is the Approximate Voronoi Diagrams (AVD).
They are shown to establish a tradeoff between the space complexity
of the data structure and the query time it supports \cite{AMM09}.
With a tradeoff parameter $2 \leq \gamma \leq \frac{1}{\epsilon}$,
the query time is
$O( \log(n\gamma ) + {1} / {(\epsilon\gamma )^{\frac{d-1}{2}}})$
and the space is $O( n \gamma^{d-1} \log \frac{1}{\epsilon})$.
They are implemented on a hierarchical quadtree-based subdivision 
of space into cells, each storing a number of representative points, such that for any query point lying in the cell, at least one 
of the representatives is an approximate nearest neighbor. Further improvements to the space-time trade offs for ANN are obtained in 
\cite{AdFM11}. 

One might apply the Johnson-Lindenstrauss Lemma and
map the points to $O({\epsilon^{-2}}{\log n})$ dimensions with distortion equal to $1+\epsilon$ 
aiming at improving complexity. In particular,
AVD combined with the Johnson-Lindenstrauss Lemma have
query time polynomial in $\log n$, $d$ and $1/\epsilon$ but require
$n^{O({\log ({1}/{\epsilon})}/{\epsilon^2})}$ space, which is prohibitive if $\epsilon\ll 1$. Notice that we 
relate the approximation error with the distortion for simplicity. 
Our approach (Theorem~\ref{thmbadnear}) requires $O(dn)$ space and has
query time 
$O(d n^{\rho})$, where $\rho \approx 1-{\epsilon^2}/{\log(1/\epsilon)}$.

In high dimensional spaces, classic space partitioning data structures
are affected by the curse of dimensionality, as illustrated above. This means that, when
the dimension increases, either the query time or
the required space increases exponentially.
An important method conceived for high dimensional data  
is locality sensitive hashing (LSH). LSH induces a data independent
random partition and is dynamic, since it supports insertions and deletions.
It relies on the existence of locality sensitive hash functions,
which are more likely to map similar objects to the same bucket.
The existence of such functions depends on the metric space. 
In general, LSH requires roughly $O(dn^{1+\rho})$ space and
$O(dn^{\rho})$ query time for some parameter $\rho \in (0,1)$. 
In \cite{AI08} they show that in the Euclidean case, 
one can have $\rho=\frac{1}{(1+\epsilon)^2}$
which matches the lower bound of hashing algorithms proved in \cite{OWZ11}.
Lately, it was shown that it is possible to overcome 
this limitation by switching to a data-dependent scheme  which achieves 
$\rho = \frac{1}{2(1+\epsilon)^2-1}+o(1)$ \cite{AR15}. 
One different approach \cite{Pan06} 
focuses on using near linear space but with query time proportional to 
$dn^{O(1/{(1+\epsilon)})}$ which is sublinear only  when $\epsilon$ is large enough. The query time was later improved  
\cite{AI08}
to $dn^{O(1/{(1+\epsilon)^2})}$ which is also sublinear only for large enough $\epsilon$.
For comparison, 
in Theorem~\ref{thmbadnear} we show that it is possible to use near linear space, 
with query time roughly $O(dn^{\rho})$, where  $\rho \approx 1-{\epsilon^2}/{\log(1/\epsilon)}$, 
achieving sublinear query time even for small values of $\epsilon$. 

Exploiting the structure of the input is an important way to improve
the complexity of ANN.
In particular, significant amount of work has been done for pointsets
with low doubling dimension.
In \cite{HPM05}, they provide an algorithm with
expected preprocessing time $O(2^{\dim(X)} n \log{n})$,
space usage $O(2^{\dim(X)} n)$ and query time
$O(2^{\dim(X)}\log{n}+\epsilon^{-O(\dim(X))})$ for any finite 
metric space $X$ of doubling dimension $\dim(X)$.
In \cite{IN07} they provide randomized embeddings that preserve
nearest neighbor with constant probability, for points lying on
low doubling dimension manifolds in Euclidean settings. Naturally, 
such an approach can be easily combined with any 
known data structure for $\epsilon$-ANN.

In \cite{DF08} they present random projection trees 
which adapt to pointsets of low doubling dimension. 
Like kd-trees, every split partitions the pointset into
subsets of roughly equal cardinality.
Unlike kd-trees, the space is split with respect to
a random direction, not necessarily parallel to the coordinate axes.
Classic $kd$-trees also adapt to the doubling dimension of 
randomly rotated data \cite{Vem12}.
However, for both techniques, no related theoretical arguments 
about the efficiency of $\epsilon$-ANN search were given.

In \cite{KR02}, they introduce a different notion of intrinsic dimension
for an arbitrary metric space, namely its expansion rate $c$;
it is formally defined in Subsection~\ref{SSexpansion}.
The doubling dimension is a more general notion of intrinsic dimension
in the sense that, when a finite metric space has 
bounded expansion rate, then it also has bounded doubling dimension,
but the converse does not hold \cite{GKL03}. Several efficient 
solutions are known for metrics with bounded expansion rate, including
for the problem of exact nearest neighbor. 
In \cite{KL04}, they present 
a data structure which requires $c^{O(1)}n$ space and answers queries
in $c^{O(1)}\ln n$. Cover Trees \cite{BKL06} require 
$O(n)$ space and each query costs $O(c^{12} \log n)$ time for exact 
nearest neighbors. In Theorem~\ref{TannExp}, 
we provide a data structure for the $\epsilon$-ANN 
problem with linear space and roughly $ O((c^{\log (\log c/\epsilon)})d \cdot \log n) )$ query time. 
The result concerns pointsets in $d$-dimensional Euclidean space.

\section{Low Quality Randomized Embeddings}\label{Sembed}

This section examines standard dimensionality reduction techniques and
extends them to approximate embeddings optimized to our setting. In the following, 
we denote by $\|\cdot \|$ the Euclidean norm and by $|\cdot|$ the cardinality of a set.

 An embedding is {\em oblivious} when it can be 
computed for any point of a dataset or query set, without knowledge of any other point in these 
sets. 

In \cite{ABetal05}, they consider non-oblivious embeddings from finite metric spaces with small dimension and 
distortion, while allowing a constant fraction of all distances to be arbitrarily distorted. In \cite{BRS11}, they present non-oblivious 
embeddings for the $\ell_2$ case, which preserve distances in local neighborhoods. 
In \cite{GK15}, they provide a non-oblivious embedding which preserves distances up to a given scale 
and the target dimension mainly depends on $ddim(X)$ with no dependence on $|X|$. In general, embeddings based on 
probabilistic partitions are not oblivious. In \cite{BG15a}, they solve ANN in 
$\ell_p$ spaces, for $2<p<\infty$, by oblivious embeddings to $\ell_{\infty}$ or $\ell_2$. 

But, it is not obvious how to use a non-oblivious embedding in the scenario in which we preprocess a dataset and we expect a 
query to arrive. Therefore we focus on oblivious embeddings.

Let us now revisit the classic Johnson-Lindenstrauss Lemma:  

\begin{prp} {\rm\cite{JL84}}
For any set $X \subset \mathbb{R}^d$, $\epsilon \in (0,1)$ there exists a distribution over linear mappings 
$f:\RR^d \longrightarrow \RR^{d'}$, where $d'=O(\log |X|/\epsilon^2)$, such that for any $p,q \in X$, 
$$
(1-\epsilon) \| p-q\|^2 \leq \|f(p)-f(q)\|^2 \leq (1+\epsilon) \|p-q\|^2.
$$
\end{prp}

In the initial proof \cite{JL84}, they show that this can 
be achieved by orthogonally projecting the pointset on a 
random linear subspace of dimension $d'$. In \cite{DG02}, they provide 
a proof based on elementary probabilistic techniques, see also Lemma~\ref{lemDG}.
In \cite{IM98}, they prove that it suffices to apply a gaussian matrix $G$ 
on the pointset. $G$ is a $d \times d'$ matrix with each of its entries 
independent random variables given by the standard normal distribution $N(0,1)$. 
Instead of a gaussian matrix, we can even apply a matrix whose entries 
are independent random variables with uniformly distributed values in $\{-1,1\}$ \cite{Ach03}. 

However, it has been realized that this notion 
of randomized embedding is stronger than what is required for ANN searching. 
The following definition has been introduced in~\cite{IN07} and focuses only on the distortion of the nearest neighbor.

\begin{dfn}  \label{Dnnpres}
Let $(Y,d_Y)$, $(Z,d_Z)$ be metric spaces and $X \subseteq Y$.
A distribution over mappings $f:Y \rightarrow Z$ is
a {\em nearest-neighbor preserving embedding} with distortion $D \geq 1$
and probability of correctness $P \in [0,1]$ if,
$\forall \epsilon \geq 0$ and $\forall q\in Y$, with probability at least $P$, 
when $x\in X$ is such that $f(x)$ is an $\epsilon$-ANN 
of $f(q)$ in $f(X)$, then $x$ is a $(D \cdot (1+\epsilon))$-approximate nearest neighbor
of $q$ in $X$. 
\end{dfn}

Let us now consider a closely related problem.
While in the ANN problem we search one point which is approximately nearest,
in the $k$ approximate nearest neighbors problem ($\epsilon$-$k$ANNs) we seek 
an approximation of the $k$ nearest points, in the following sense.
 Let $X$ be a set of $n$ points in $\mathbb{R}^d$, let $q\in \RR^d$ and $1\leq k \leq n$.
The problem consists in reporting a sequence $S=\{p_1,\dots,p_k\}$ of $k$ distinct points such that the $i$-th point $p_i$
is an $(1+\epsilon)$-approximation to the $i$-th nearest neighbor of $q$. Furthermore, the following assumption 
is satisfied by the search routine of certain tree-based data structures, such as BBD-trees. 

\begin{ass}\label{assbbd}
 Let $S'\subseteq X$ be the 
set of points visited by the $\epsilon$-$k$ANNs search such that 
$S=\{p_1,\dots,p_k\} \subseteq S'$ is the set of points which are the $k$ nearest points
to the query point $q$ among the points in $S'$.
Moreover, $S$ is ordered w.r.t.\ distance from $q$, hence $p_k$ is farthest.
We assume that 
$\forall x \in X \setminus S'$, $d(x,q) > d(p_k,q)/(1+\epsilon)$. 
\end{ass}

Assuming the existence of a data structure which solves $\epsilon$-$k$ANNs and satisfies Assumption~\ref{assbbd}, 
we propose to weaken Definition~\ref{Dnnpres} as in the following definition. 

\begin{dfn} \label{Dlocpres}
Let $(Y,d_Y)$, $(Z,d_Z)$ be metric spaces and $X \subseteq Y$.
A distribution over mappings 
$f:Y \rightarrow Z$ is a {\em locality preserving embedding}
with distortion $D \geq 1$, probability of correctness 
$P \in [0,1]$ and locality parameter $k$ if, $\forall \epsilon \geq 0$ and
$\forall q\in Y$, with probability at least $P$, 
when $S=\{f(p_1),\dots,f(p_k)\}$ is a solution to $\epsilon$-$k$ANNs for $q$
under Assumption~\ref{assbbd}, then there exists $f(x) \in S$ such that $x$ is a $(D\cdot (1+\epsilon))$-approximate nearest neighbor of $q$ in $X$.
\end{dfn}

According to this definition we can reduce the problem of $\epsilon$-ANN in
dimension $d$ to the problem of computing $k$ approximate nearest neighbors
in dimension $d'<d$.

We employ the Johnson-Lindenstrauss dimensionality reduction technique and, more specifically, the proof in \cite{DG02}. 

\begin{rmk}
 In the statements of our results, we use the term $(1+\epsilon)^2$ or $(1+\epsilon)^3$ for the sake of simplicity. 
 Notice that we can replace $(1+\epsilon')^2$ by $1+\epsilon$ just by 
 rescaling $\epsilon' \gets \epsilon/4\implies (1+\epsilon')^2\leq 1+\epsilon$.
\end{rmk}

\begin{lem} {\rm\cite{DG02}} \label{lemDG}
There exists a distribution over linear maps
$A:\mathbb{R}^d \rightarrow \mathbb{R}^{d'}$ s.t., for any $p \in \mathbb{R}^d$ with $\|p\|=1$:
\begin{itemize}
 \item if $\beta^2<1$ then $\mathrm{Pr}[\|Ap\|^2 \leq \beta^2 \cdot \frac {d'}{d}] \leq exp(\frac{d'}{2}(1-\beta^2+2 \ln \beta)),$
 \item if $\beta^2>1$ then $\mathrm{Pr}[\|Ap\|^2 \geq \beta^2 \cdot \frac {d'}{d}] \leq exp(\frac{d'}{2}(1-\beta^2+2 \ln \beta)).$
 \end{itemize}
\end{lem}

Now, a simple calculation shows the following.
\begin{cor}
 \label{corDG}
 If $\beta^2<1$ then $\mathrm{Pr}[\|Ap\|^2 \leq \beta^2 \cdot \frac {d'}{d}] \leq (e\beta^2)^{d'/2}$.
\end{cor}
\begin{proof}
By Lemma~\ref{lemDG},
 $$\mathrm{Pr}[ \|Ap\|^2 \leq \beta^2 \cdot \frac {d'}{d} ] \leq ({e^{-\beta^2+1+2 \ln \beta}})^{d'/2} \leq ({e^{1+2 \ln \beta}})^{d'/2}.
 $$
\end{proof}

The following inequality shall be useful.

\begin{lem} \label{LemIneq}
 For all $i\in \mathbb{N}$, $\epsilon \in (0,1/2]$, the following holds:
$$
\frac{1}{(1+\epsilon)^2} -
2 \ln \frac{1}{1+\epsilon}-1 >  \epsilon^2.
$$
\end{lem}
\begin{proof}
 Let $f(\epsilon)=\frac{1}{(1+\epsilon)^2} -
2 \ln \frac{1}{(1+\epsilon)}-1 -  \epsilon^2$, which is continuous in $(0,1/2]$. 
It suffices to show that $f(\epsilon)>0$, for $\epsilon \in (0,1/2]$. 
Then we examine its derivative:
$$
f'(\epsilon)=2 \epsilon \Big( \frac{2+\epsilon}{(1+\epsilon)^3}-1\Big).
$$
Since $\epsilon>0$, we need to examine $\frac{2+\epsilon}{(1+\epsilon)^3}-1$. We have,
$$\frac{2+\epsilon}{(1+\epsilon)^3}-1 \geq 0  \iff {2+\epsilon}\geq{(1+\epsilon)^3} \iff \epsilon^3+3\epsilon^2+2\epsilon-1 \leq 0 .
$$
The last inequality holds when $\epsilon \leq z$, where $z \approx 0.32$. while $f'(\epsilon)<0$ for $\epsilon \in (z,1/2]$. 
Hence, $f(\epsilon)$ is an increasing function when $\epsilon \in (0,z]$ and decreasing in $[z,1/2]$. 
Now, in the interval $\epsilon \in (0,z]$ we obtain $f(\epsilon)>f(0)=0$ and in 
$\epsilon \in [z,1/2]$ we obtain $f(\epsilon)\geq f(1/2)>0.005$.

\end{proof}

We are now ready to prove the main theorem of this section.

\begin{thm} \label{thmbad}
Under the notation of Definition~\ref{Dlocpres},
there exists a randomized mapping $f:\mathbb{R}^d \rightarrow \mathbb{R}^{d'}$
which satisfies Definition~\ref{Dlocpres} for
$d'=O({\log \frac{n}{ k}} / {\epsilon^2})$, $\epsilon\in(0,1/2]$, 
distortion $D=(1+\epsilon)^2$ and probability of success $2/3$.
\end{thm}

\begin{proof} 
Let $X$ be a set of $n$ points in $\mathbb{R}^d$ and consider map
$$
f:\mathbb{R}^d \rightarrow \mathbb{R}^{d'} : v\mapsto \sqrt{d/{d'}}\cdot A\ v,
$$
where $A$ is a matrix chosen from a distribution as in Lemma~\ref{lemDG}.
Without loss of generality the query point $q$ lies at the origin and its 
nearest neighbor $u$ lies at distance $1$ from $q$. We denote by $c\geq 1$ 
the approximation ratio guaranteed by the assumed data structure (see Assumption~\ref{assbbd}). That is, the assumed data structure 
solves the $(c-1)$-$k$ANNs problem. 
Let $N$ be the random variable whose value indicates the number
of ``bad'' candidates, that is 
$$
N= |\, \{ x \in X \, :\, \|x\|>\gamma \, \wedge\,
	{\|f(x)\|} \leq \beta \}\, |,
$$ 
where we define $\beta=c(1+\epsilon)$, $\gamma=c(1+\epsilon)^2$.
Hence, by Lemma~\ref{lemDG} and Lemma~\ref{LemIneq}, 
$$
\mathbb{E}[N]\leq n \cdot exp(\frac{d'}{2}(1-\frac{\beta^2}{\gamma^2}+2 \ln \frac{\beta}{\gamma}))\leq {n} \cdot  exp(-d' \cdot \epsilon^2/2).
$$ 
The event of failure is defined as the disjunction of two events:
\begin{equation} \label{Eproba}
 N\geq k\,  \; \vee \; \, \|f(u)\| \geq (\beta/c) ,
\end{equation}
and its probability is at most equal to
$$
\mathrm{Pr}[N \geq k] + exp(\frac{d'}{2}(1-(\beta/c)^2+2 \ln (\beta/c))) ,
$$
by applying again Lemma~\ref{lemDG}.  Now, we set $d'\geq 2 \ln (\frac{6n}{k})/\epsilon^2$ and we bound these two terms. %
By Markov's inequality,
$$
\mathrm{Pr}[N\geq k] \leq \frac{n}{k} \cdot  exp(-d' \cdot \epsilon^2/2) \leq\frac{1}{6}.
$$
In addition,
$$
exp(\frac{d'}{2}(1-(\beta/c)^2+2 \ln (\beta/c))){\leq}\Big(\frac{k}{6n}\Big)^{2-2\epsilon/3} <\frac{1}{6}.
$$
Hence, there exists $d'$ such that
$$
d'=O({\log \frac{n}{ k}} / {\epsilon^2})
$$
and with probability at least $2/3$, the following two events occur:
$$ \|f(q)-f(u)\|\leq (1+\epsilon) \|u-q\|,$$
$$ | \{ p \in X | \|p-q\| > c(1+\epsilon)^2 \|u-q\| \implies \|f(q)-f(p)\|\leq c(1+{\epsilon}) \|u-q\| \}| <k.$$

Let us consider the case when the random experiment succeeds, and 
let $S=\{f(p_1),\dots ,f(p_k)\}$ be a solution of the $(c-1)$-$k$ANNs problem in the projected space, 
given by a data-structure which satisfies Assumption~\ref{assbbd}. 
It holds that $\forall f(x) \in f(X) \setminus S'$, $\|f(x)-f(q)\|> \|f(p_k)-f(q)\|/c$, 
where $S'$ is the set of all points visited by the search routine.  

If $f(u) \in S$, then $S$ contains the projection of the nearest neighbor. If $f(u) \notin S$, then if $f(u) \notin S'$ 
we have the following:
$$
\|f(u)-f(q)\|>\|f(p_k)-f(q)\|/c \implies \|f(p_k)-f(q)\| < c(1+\epsilon)\|u-q\|,
$$
which means that there exists at least one point $f(p^{*})\in S$ s.t.\ $\|q-p^{*}\| \leq c(1+\epsilon)\|u-q\|$. Finally, if 
$f(u) \notin S$ but $f(u) \in S'$ then 
$$\|f(p_k)-f(q)\| \leq \|f(u)-f(q)\| \implies  \|f(p_k)-f(q)\| \leq (1+\epsilon) \|u-q\|,$$
which means that there exists at least one point $f(p^{*})\in S$ s.t. $\|q-p^{*}\| \leq c(1+\epsilon)^2\|u-q\|$. 

Hence, $f$ satisfies Definition~\ref{Dlocpres} for $D=(1+\epsilon)^2$ and the theorem is established. 
\end{proof}

\section{Approximate Nearest Neighbor Search}\label{Sann}

This section combines tree-based data structures which solve $\epsilon$-$k$ANNs with the results of Section~\ref{Sembed},
in order to obtain an efficient randomized data structure which solves $\epsilon$-ANN.

\subsection{Finite subsets of $\ell_2$}\label{SSannEu}
This subsection examines the general case of finite subsets of $\ell_2$. 

BBD-trees \cite{AMN+98} require $O(dn)$ space, and allow computing $k$
points, which are $(1+\epsilon)$-approximate nearest neighbors, in
time $O((\lceil 1+6\frac{d}{\epsilon}\rceil^d+k) d \log n)$.
The preprocessing time is $O(dn \log n)$. Notice, that BBD-trees satisfy 
Assumption~\ref{assbbd}. 

The algorithm for the $\epsilon$-$k$ANNs 
search visits cells in increasing order with respect to their distance from the query point $q$. If the current cell 
lies at distance more than $r_k/c$, where $r_k$ is the current distance to the $k$th nearest neighbor, the search terminates. 
We apply the random projection for distortion $D=1+\epsilon$, thus
relating approximation error to the allowed distortion; this is not 
required but simplifies the analysis.

Moreover, $k=n^{\rho}$; the formula for $\rho<1$ is determined below.
Our analysis then focuses on the asymptotic behavior of the term 
$O( \lceil 1+6\frac{d'}{\epsilon}\rceil^{d'} + k )$.  

\begin{lem} \label{lemBBD}
With the above notation, there exists $k>0$ s.t., for fixed
$\epsilon \in (0,1)$, it holds that
$\lceil 1+6\frac{d'}{\epsilon}\rceil^{d'}+k =O(n^{\rho})$, where 
$\rho = 1-\Theta({\epsilon^2}/{(\epsilon^2+\log(\max\{\frac{1}{\epsilon},\log n\}))})<1$.
\end{lem}

\begin{proof} 
Recall that $d'\leq \frac{\tilde{c}}{\epsilon^2} \ln \frac{n}{k}$
for some appropriate constant $\tilde{c} >0$. 
Since $(\frac{d'}{\epsilon})^{d'}$ 
is a decreasing function of $k$, we need to choose $k$ s.t.\
$(\frac{d'}{\epsilon})^{d'} = \Theta(k)$. 
Let $k=n^\rho$. 
It is easy to see that
$\lceil 1+6\frac{d'}{\epsilon}\rceil^{d'}\leq (c'\frac{d'}{\epsilon})^{d'}$,
for some appropriate constant $c' \in (1,7)$.  
Then, by substituting $d',k$ we obtain: 
\begin{equation}\label{Eterm1ne}
\ln (c'\frac{d'}{\epsilon})^{d'}= { \frac{\tilde{c} (1-\rho)}{\epsilon^2}\ln ( \frac{\tilde{c} c'(1-\rho) \ln n}{\epsilon^3} )} \ln n.
\end{equation}

We assume $\epsilon \in (0,1)$ is a fixed constant. Hence, it is reasonable to assume that $\frac{1}{\epsilon}<n$.
We consider two cases when comparing $\ln n$ to $\epsilon$:
\begin{itemize}
\item
$\frac{1}{\epsilon} \leq \ln n$. Substituting $\rho=1-\frac{\epsilon^2}{2\tilde{c}(\epsilon^2+\ln ( c' \ln n))}$ into 
equation (\ref{Eterm1ne}), the exponent of $n$ is bounded as follows:
$$\frac{\tilde{c} (1-\rho)}{\epsilon^2}
\ln ( \frac{\tilde{c} c'(1-\rho) \ln n}{\epsilon^3} ) =
$$
$$={\frac{\tilde{c}  }{ 2\tilde{c}(\epsilon^2+\ln( c' \ln n))} \cdot 
\big( \ln (c' \ln n) + \ln{\frac{1}{\epsilon}} -\ln{(2\epsilon^2+2\ln (c' \ln n)})\big)}
<{\rho}.
$$
 \item
$\frac{1}{\epsilon}> \ln n$.
Substituting $\rho=1-\frac{\epsilon^2}{2\tilde{c}(\epsilon^2+\ln  \frac{c'}{\epsilon})}$ into equation (\ref{Eterm1ne}), the exponent of 
$n$ is bounded as follows:
$$\frac{\tilde{c} (1-\rho)}{\epsilon^2}
\ln ( \frac{\tilde{c} c'(1-\rho) \ln n}{\epsilon^3} ) = $$
$$
=  {\frac{\tilde{c}  }{ 2\tilde{c}(\epsilon^2+\ln \frac{c'}{\epsilon}) } \cdot
\big(  \ln  \ln n + \ln{\frac{c'}{\epsilon}} -\ln{(2\epsilon^2+2\ln \frac{c'}{\epsilon}})\big)}<  
{\rho}.
$$
\end{itemize}
\end{proof}

Notice that in both cases 
$$
d'=O(\frac{\log n}{\epsilon^2+\log \log n}).
$$

Combining Theorem~\ref{thmbad} with Lemma~\ref{lemBBD}
yields the following main theorem.

\begin{thm}\label{Tann}
Given $n$ points in $\RR^d$,
there exists a randomized data structure which requires $O(dn)$ space
and reports an $(1+\epsilon)^3$-approximate nearest neighbor in 
time $$O(d n^{\rho} \log n),\text{ where }
\rho \leq 1-\Theta({\epsilon^2}/{(\epsilon^2+\log(\max\{\frac{1}{\epsilon},\log n\}))}<1$$ 
The preprocessing time is $O(d n \log n)$.
For each query $q \in \mathbb{R}^d$, the preprocessing phase 
succeeds with any constant probability.
\end{thm}
\begin{proof}
The space required to store the dataset is $O(dn)$. The space used 
by BBD-trees is $O(d'n)$ where $d'$ is defined in Lemma~\ref{lemBBD}. 
We also need $O(dd')$ space for the matrix $A$ as specified in Theorem ~\ref{thmbad}. 
Hence, since $d'<d$ and $d'<n$, the total space usage is bounded above 
by $O(dn)$.
 
The preprocessing consists of building the BBD-tree which costs $O(d'n \log n)$ time 
and sampling $A$. Notice that we can sample a $d'$-dimensional 
random subspace in time $O(d d'^2)$ as follows. First, we sample in time $O(d d')$, 
a $d \times d'$ matrix where its elements are independent random variables 
with the standard normal distribution $N (0, 1)$. 
Then, we orthonormalize using Gram-Schmidt in time $O(d d'^2)$. Since $d'=O( \log n)$, 
the total preprocessing time is bounded by $O(d n \log n)$. 

For each query we use $A$ to project the point in time $O(d d')$. Next, we compute its 
$k=n^{\rho}$ approximate nearest neighbors in time $O(d' n^{\rho} \log n)$ and we check these neighbors with their 
$d$-dimensional coordinates in time $O(d n^{\rho})$. Hence, each query costs $O(d \log n+d' n^{\rho}\log n+dn^{\rho})=O(d n^{\rho} \log n)$ 
because $d'=O(\log n)$, $d'=O(d)$. Thus, the query time is dominated by the time required for $\epsilon$-$k$ANNs search and the time 
to check the returned sequence of $k$ approximate nearest neighbors.
\end{proof}

To be more precise, the probability of success, which is the probability that 
the random projection succeeds according to 
Theorem.~\ref{thmbad}, is at least constant and can be amplified to high probability of success 
with repetition. 
Notice that the preprocessing time for BBD-trees has no dependence on $\epsilon$.

\subsection{Finite subsets of $\ell_2$ with bounded expansion rate}\label{SSexpansion}

This subsection models some structure that the data points may have so as to obtain tighter bounds.

The bound on the dimension $d'$ obtained in Theorem~\ref{thmbad}
is quite pessimistic.
We expect that, in practice, the space dimension needed in 
order to have a sufficiently good projection is less than what
Theorem~\ref{thmbad} guarantees. Intuitively, we do not expect to have 
instances where all points in $X$, which are not
approximate nearest neighbors of $q$, lie at distance 
$\approx(1+\epsilon)d(q,X)$.
To this end, we consider the case of pointsets with bounded expansion rate.

\begin{dfn}
Let $M$ be a metric space and $X\subseteq M$ be a finite pointset and 
let $B_p(r)\subseteq X$ denote the points of $X$ lying in the closed
ball centered at $p$ with radius $r$. 
We say that $X$ has $({\tau},{c})$-expansion rate if and only if,
$\forall p \in M$ and $r>0$,
$$
|B_p(r)|\geq \tau \implies |B_p(2r)|\leq c \cdot |B_p(r)|.
$$
\end{dfn}

\begin{thm} \label{thmbad2}
Under the notation of Definition~\ref{Dlocpres}, 
there exists a randomized mapping $f:\mathbb{R}^d \rightarrow \mathbb{R}^{d'}$
which satisfies Definition~\ref{Dlocpres} for dimension
$d'=O(\log c)$,
distortion $D=(1+\epsilon)^2$ and constant probability of success, 
for pointsets with $(\tau,c)$-expansion rate.
 \end{thm}

\begin{proof}
We proceed in the same spirit as in the proof of Theorem~\ref{thmbad}. 

Let $X$ be a set of $n$ points in $\mathbb{R}^d$ and consider map
$$
f:\mathbb{R}^d \rightarrow \mathbb{R}^{d'} : v\mapsto \sqrt{d/{d'}}\cdot A\ v,
$$
where $A$ is a matrix chosen from a distribution as in Lemma~\ref{lemDG}.
Without loss of generality the query point $q$ lies at the origin and its 
nearest neighbor $u$ lies at distance $1$ from $q$. 
Let $r_0$ be the distance to the $\tau-$th nearest neighbor,
excluding neighbors at distance $\leq (1+\epsilon)^2$. 
For $i>0$, let $r_i=2 \cdot r_{i-1}$ and set $r_{0}=(1+\epsilon)^2$ (since, $r_0\geq (1+\epsilon)^2$). 

We distinguish the set of bad candidates according to whether they correspond to 
``close'' of ``far'' points in the initial space. More precisely,
$$
N_{close}= |\, \{ x \in X \, :\, \|x\| \in [r_0,r_1) \, \wedge\,
	 \|f(x)\| \leq {\beta}  \}\, |,
$$ 
$$
N_{far}= |\, \{ x \in X \, :\, \|x\|\geq r_1 \, \wedge\,
	 \|f(x)\| \leq \beta \}\, |,
$$ 
where $\beta=1+\epsilon$. 
Clearly, by Lemma~\ref{LemIneq}, and for $d'\geq \ln c +1$,

$$
\EE [N_{close}] \leq c\cdot \tau\cdot exp(-d'\cdot \epsilon^2/2) =O(c \cdot \tau),
$$
and similarly by Corollary~\ref{corDG}, 
$$
\EE[N_{far}]\leq 
\sum_{i=1}^{\infty} c^{i+1} \tau \cdot \Big(\frac{e (1+\epsilon)^2}{r_{i}^2}\Big)^{d'/2} 
\leq \sum_{i=1}^{\infty} c^{i+1} \tau \cdot \Big(\frac{e }{2^{2i}}\Big)^{d'/2}=
e^{d'/2} \cdot \tau \cdot c\sum_{i=1}^{\infty} c^{i}  \Big(\frac{1 }{2^{i}}\Big)^{d'}= O(\tau \cdot c^2).
$$
Finally, 
using Markov's inequality, we obtain constant probability of success.
\end{proof}

Employing Theorem~\ref{thmbad2} we obtain a result analogous to Theorem~\ref{Tann} which is
weaker than those in \cite{KL04,BKL06} but underlines the fact that
our scheme shall be sensitive to structure in the input data, for
real world assumptions.

\begin{thm}\label{TannExp}
Given $n$ points in $\ell_2^d$ 
with $(\tau,c)$-expansion rate, for some constant $c$,
there exists a randomized data structure which requires $O(dn)$ space
and reports an $(1+\epsilon)^3$-approximate nearest neighbor in time
$$
O((c^{\log (\log c/\epsilon)}+ \tau \cdot c^2\, )\, d \, \log n) ).
$$
The preprocessing time is $O(d n \log n)$.
For each query $q \in \mathbb{R}^d$, the preprocessing phase 
succeeds with constant probability.
\end{thm}

\begin{proof}
We combine the embedding of Theorem~\ref{thmbad2} with the BBD-trees. Then,
$$O\Big(\Big(\frac{\sqrt{d'}}{\epsilon}\Big)^{d'}\Big)=O\Big(\Big(\frac{\log c}{\epsilon}\Big)^{\log c}\Big)$$
and the number of approximate nearest neighbors in the projected space is
$$
k=O(\tau \cdot c^2).
$$
This proves the result.
\end{proof}

\section{Approximate Near Neighbor}\label{Snear}

This section combines the ideas developed in  Section~\ref{Sembed} with a simple, auxiliary data structure, namely the grid, yielding an efficient solution for the 
$(\epsilon,R)$-ANN problem.

\paragraph{Problem Definition}
Building a data structure for the Approximate Nearest Neighbor Problem reduces to 
building several data structures for the $(\epsilon,R)$-ANN Problem. For completeness, 
we include the corresponding theorem.

\begin{thm}{\rm\cite[Thm~2.9]{HIM12}}
Let $P$ be a given set of $n$ points in a metric space, and let $c = 1 + \epsilon > 1$, $f \in (0, 1)$, and
$\gamma \in (1/n, 1)$ be prescribed parameters.
Assume that we are given a data structure for the $(c, r)$-approximate near neighbor that uses space
$S(n, c, f )$, has query time $Q(n, c, f )$, and has failure probability $f $. Then there exists a data structure for
answering $c(1 + O(\gamma))$-NN queries in time $O(\log n)Q(n, c, f )$ with failure probability $O( f \log n)$. The
resulting data structure uses $O(S(n, c, f )/\gamma \cdot \log^2 n)$ space.
\end{thm}

In the following, the $\tilde{O} $ notation hides 
factors polynomial in $1/\epsilon$ and $\log n$.
When the dimension is high the problem has been solved efficiently by randomized methods based on the notion of LSH.
\begin{dfn}[$(c,R)$-ANN Problem (as studied in the high dimensional case)]
\label{annhigh}
 Let $X\subseteq \mathbb{R}^d$ and $|X|=n$. Given $\epsilon>0,r>0$, build a data structure 
 which for any query $q\in \mathbb{R}^d$ the probability that the building phase of the data structure succeeds for $q$ 
 is at least constant. 
\end{dfn}

A natural generalization of the $(\epsilon,R)$-ANN problem is the $k$-Approximate Near Neighbors Problem ($(\epsilon,R)$-$k$ANNs).

\begin{dfn}[$(\epsilon,R)$-$k$ANNs Problem]
\label{DfnkAnns}
 Let $X\subset \mathbb{R}^d$ and $|X|=n$. Given $\epsilon>0$, $R>0$, build a data structure 
 which, for any query $q\in \mathbb{R}^d$:
 \begin{itemize}
  \item if $|\{p \in X \mid \|q-p\|\leq R \}|\geq k$, then report $S \subseteq \{p \in X \mid \|q-p\|\leq (1+\epsilon)R \}$ s.t.\ $|S|= k$,
   \item if $|\{p \in X \mid \|q-p\|\leq R \}|<k$, then report $S \subseteq \{p \in X \mid \|q-p\|\leq (1+\epsilon)R \}$  s.t.\ 
     $|\{p \in X \mid \|q-p\|\leq R \}|\leq |S|\leq k$.
 \end{itemize}
\end{dfn}

The following algorithm is essentially the bucketing method which is described in~\cite{HIM12} and concerns the case $k=1$.
Impose a uniform grid of side length $\epsilon/\sqrt{d}$ on $\RR^d$. 
Clearly, the
distance between any two points belonging to one grid cell is at most $\epsilon$. Assume $r=1$. 
For each ball $B_q = \{x\in \RR^d \mid \|x-q\|\leq r\}$,
$q\in \RR^d$, let ${\overline{B_q}}$ be the set of grid cells that intersect $B_q$.   

In~\cite{HIM12}, 
they show that $|\overline{B_q}|\leq (C/\epsilon)^d$. 
Hence, the query time is the time to compute the hash function, 
retrieve near cells and report the $k$ neighbors: $$O(d+k+(C/\epsilon)^d).$$ 
The required space usage is $O(dn)$.

Furthermore, we are interested in optimizing this constant $C$. The bound on $|\overline{B_q}|$ follows from the following fact:
$$
|\overline{B_q}| \leq V_2^d(R),
$$
where $V_2^d(R)$ is the volume of the ball with radius $R$ in $\ell_2^d$,
and $R=\frac{2 \sqrt{d}}{\epsilon}$. Now, 
$$
V_2^d(R) \leq \frac{2\pi^{d/2}}{d \cdot\Gamma(d/2)} R^d=\frac{2\pi^{d/2}}{d(d/2-1)!} R^d \leq \frac{2\pi^{d/2}}{(d/2)!} R^d \leq 
 \frac{2\pi^{d/2}}{e(d/(2e))^{d/2}} R^d \leq \frac{2^{d+1}(18)^{d/2}}{ \epsilon^d e} \leq \frac{9^{d}}{ \epsilon^d } .
$$
Hence, $C\leq 9$.

\begin{thm}\label{ThmkAnns}
 There exists a data structure for the Problem \ref{DfnkAnns} with required space 
 $O( dn) $ and query time $O(d+k + (\frac{C}{\epsilon})^d))$, for $C\leq 9$.
\end{thm}

The following theorem is an analogue of Theorem~\ref{thmbad} for the Approximate Near Neighbor Problem.
\begin{thm}
 The $((1+\epsilon)^2c,R)$-ANN problem in $\mathbb{R}^d$ reduces to checking the solution set of the 
 $(c,(1+\epsilon)R)$-$k$ANNs problem in $\mathbb{R}^{d'}$, where 
 $d'=O({\log (\frac{n}{k})}/{\epsilon^2})$, by a randomized algorithm which succeeds with constant probability. The delay in query time is proportional to $d\cdot k$. 
\end{thm}

\begin{proof} 
The theorem can be seen as a direct implication of Theorem~\ref{thmbad}. The proof is indeed the same. 

Let $X$ be a set of $n$ points in $\mathbb{R}^d$ and consider map
$$
f:\mathbb{R}^d \rightarrow \mathbb{R}^{d'} : v\mapsto \sqrt{d/{d'}}\cdot A\ v,
$$
where $A$ is a matrix chosen from a distribution as in Lemma~\ref{lemDG}.
Let $u\in X$ a point at distance $1$ from $q$ and assume without loss of generality that lies at the origin.
Let $N$ be the random variable whose value indicates the number
of ``bad'' candidates, that is 
$$
N= |\, \{ x \in X \, :\, \|x\|>\gamma \, \wedge\,
	 \|f(x)\| \leq \beta \}\, |,
$$ 
where we define $\beta=c(1+\epsilon)$, $\gamma=c(1+\epsilon)^2$.
Hence, by Lemma~\ref{lemDG} and Lemma~\ref{LemIneq}, 
$$
\mathbb{E}[N]\leq n \cdot exp(\frac{d'}{2}(1-\frac{\beta^2}{\gamma^2}+2 \ln \frac{\beta}{\gamma}))\leq {n} \cdot  exp(-d' \cdot \epsilon^2/2).
$$ 
The probability of failure is at most equal to
$$
\mathrm{Pr}[N \geq k] + exp(\frac{d'}{2}(1-(\beta/c)^2+2 \ln (\beta/c))) ,
$$
by applying again Lemma~\ref{lemDG}.  Now, we bound these two terms. 
By Markov's inequality,
$$
\mathrm{Pr}[N\geq k] \leq \frac{n}{k} \cdot  exp(-d' \cdot \epsilon^2/2) \stackrel{d'\geq 2 \ln (\frac{6n}{k})/\epsilon^2}{\leq}\frac{1}{6}.
$$

In addition,
$$
exp(\frac{d'}{2}(1-(\beta/c)^2+2 \ln (\beta/c)))\stackrel{d'\geq 2 \ln (\frac{6n}{k})/\epsilon^2}{\leq}\Big(\frac{k}{6n}\Big)^{2-2\epsilon/3} <\frac{1}{6}
$$

Hence, there exists $d'$ such that
$$
d'=O({\log \frac{n}{ k}} / {\epsilon^2})
$$
and with probability at least $2/3$, these two events occur:
\begin{itemize}
\item $ \|f(q)-f(u)\|\leq (1+\epsilon) .$
\item $ | \{ p \in X | \|p-q\| > c(1+\epsilon)^2  \implies \|f(q)-f(p)\|\leq c(1+{\epsilon}) \}| <k.$
\end{itemize}
\end{proof}

\subsection{Finite subsets of $\ell_2$}\label{SSnear}
We are about to show what Theorem \ref{thmbadnear} implies for the data structure from Theorem \ref{ThmkAnns}.

\begin{thm}
 There exists a data structure for the Problem \ref{annhigh} with ${O}(dn)$ required space and preprocessing time,  
and query time $\tilde{O}(d n^{\rho} )$
, where $\rho=1-\Theta({\epsilon^2}/({\log (1/\epsilon)+\epsilon^2}))<1$.
\label{thmbadnear}
\end{thm}
\begin{proof}
 $$\left(\frac{C}{\epsilon}\right)^{d'} \leq \left(\frac{C}{\epsilon}\right)^{20\ln \frac{20n}{k}/\epsilon^2}=\left(\frac{20n}{k}\right)^{20\ln\frac{C}{\epsilon}/\epsilon^2} 
$$
and for $$k\geq 20n^{1-{\epsilon^2}/({20\ln (C/\epsilon)+\epsilon^2})} \implies \left(\frac{C}{\epsilon}\right)^{d'} \leq   n^{1-{\epsilon^2}/({20\ln (C/\epsilon)+\epsilon^2})}.$$

Since, the data structure succeeds only with 
probability $9/10$, it suffices to build it $O({\log n})$ times in order 
to achieve high probability of success.
\end{proof}

\subsection{The case of doubling subsets of $\ell_2$}\label{SDdembed}
In this section, we generalize the idea from \cite{AEP15} for pointsets with bounded doubling dimension to 
obtain non-linear randomized embeddings for the $(\epsilon,R)$-ANN problem. 
\begin{dfn}
The doubling dimension of a metric space $M$ is the smallest positive integer $ddim(M)$ 
such that every set $S$ with diameter $D_S$ can be covered by $2^{ddim(M)}$ (the doubling constant) sets of diameter $D_S/2$. 
\end{dfn}

Now, let $X\subset \RR^d$ s.t. $|X|=n$ and $X$ has doubling constant $\lambda_X=2^{ddim(X)}$. Consider also $S_i \subseteq X$ with diameter $2r_i$. 
Then we need $\lambda_X^{\log \frac{8 r_i}{\epsilon}}$ tiny balls $b_\epsilon\subseteq X$ of diameter $\epsilon/4$ 
in order to cover $S_i$. 
We can assume that $R=1$, since we can scale $X$. The idea is that we first compute $X'\subseteq X$ which satisfies the following two properties:
\begin{itemize}
 \item $\forall p,q \in X'$ $\|p-q\|>\epsilon/8$,
 \item $\forall q\in X ~\exists p \in X'$ s.t. $\|p-q\| \leq \epsilon/8$.
\end{itemize}
This is an $r$-net for $X$ for $r=\epsilon/8$. 
The obvious naive algorithm computes $X'$ in $O(n^2)$ time. Better algorithms exist for the case 
of low dimensional Euclidean space \cite{H04}. 
Approximate $r$-nets can be also computed in time 
$2^{O(ddim(X))}n \log n$ for doubling metrics \cite{HPM05}
, assuming 
that the distance can be computed in constant time.

Then, for $X'$ we know that each $S_i\subseteq X'$ contains $\leq \lambda_X^{\log \frac{8r_i}{\epsilon}}$ points, since 
$X'\subseteq X \implies ddim(X')\leq ddim(X)$.

\begin{thm}
 The $((1+\epsilon)^2c,R)$-ANN problem in $\mathbb{R}^d$ reduces to checking the solution set of 
 the $(c,(1+\epsilon)R)$-$k$ANNs problem in $\mathbb{R}^{d'}$, where 
 $d'=O(ddim(X))$ and $k=(2/\epsilon)^{O(ddim(X)}$, by a randomized algorithm which succeeds with constant 
 probability. Preprocessing costs 
 an additional of $O(n^2)$ time and the delay in query time is proportional to $d\cdot k$. 
\end{thm}

\begin{proof}
Once again we proceed in the same spirit as in the proof of Theorem~\ref{thmbad}. 

Let $X'$ be an $\epsilon/8$-net of $X$.  
Let $r_i={2}^{i+1}(1+\epsilon)$ for $i\geq0$ and let
$B_p(r)\subseteq X'$ denote the points of $X'$ lying in the closed
ball centered at $p$ with radius $r$. We assume $0<\epsilon\leq 1/2$. 
We make use of Corollary \ref{corDG}.
$$
\mathbb{E}[N_{far}]  \leq \;
\sum_{i=2}^{\infty} |B_p(r_i)|  \cdot (\sqrt{e} \cdot \frac{(1+\epsilon)}{r_{i-1}})^{d'}
 \leq \;
\sum_{i=2}^{\infty}  \lambda_X^{\log (8r_i/\epsilon)}
\cdot(\frac{\sqrt{e}}{2^{i}})^{d'} \leq   \lambda_X^{1+\log(16/\epsilon)} \cdot e^{d'/2} \sum_{i=2}^{\infty}  \frac{\lambda_X^i}{2^{i \cdot d'}}
=
$$
$$
\stackrel{d'\geq 1+\log \lambda_X}{=}2^{O(ddim(X) \log(2/\epsilon))}=\Big(\frac{2}{\epsilon}\Big)^{O(ddim(X))}
$$
In addition,
$$ \EE[N_{close}]\leq \lambda_X^{O(\log(1/\epsilon))} \cdot exp(-d' \cdot \epsilon^2 /20 )\leq \lambda_X^{O(\log(1/\epsilon))}=
\Big(\frac{2}{\epsilon}\Big)^{O(ddim(X))}.$$
The number of grid cells of sidewidth $\epsilon/\sqrt{d'}$ intersected by a ball of radius $1$ in $\RR^{d'}$ 
is also $(\frac{2}{\epsilon})^{O(ddim(X))}$.
Notice, that if there exists 
a point in $X$ which lies at distance $1$ from $q$, then 
there exists a point in $X'$ which lies at distance $1+\epsilon/8$ from $q$. 
Finally the probability that the distance between the query point $q$ and  
one approximate near neighbor gets arbitrarily expanded is less than 
${\lambda_X^{-\Theta(\epsilon^2)}}$.
\end{proof}

Now using the above ideas we obtain a data structure for the $(\epsilon,R)$-ANN problem. 
\begin{thm}
There exists a data structure which solves the approximate nearest neighbor problem which requires 
space and preprocessing time ${O}(d n)$ 
and the query costs 
$$d \Big(\frac{2}{\epsilon} \Big)^{O(ddim(X))}   .$$ 
For fixed $q\in \RR^d$, the building process of the data structure succeeds with constant
probability.
\end{thm}

\begin{figure}[ht] 
\centerline{\includegraphics[scale=0.42]{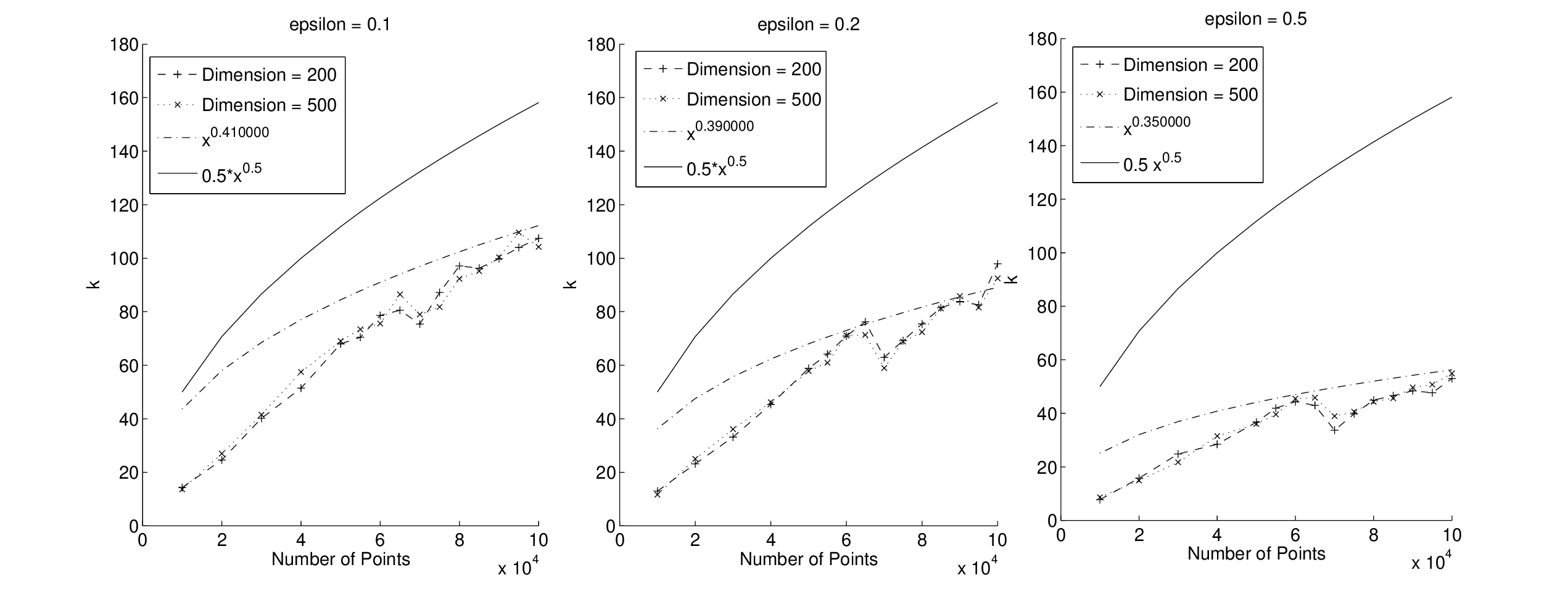}}
\caption{Plot of $k$ as $n$ increases for the ``planted nearest neighbor model'' datasets. 
The highest line corresponds to $\frac{\sqrt{n}}{2}$ 
and the dotted line to a function of the form $n^\rho,\;\text{where}\;\rho = 0.41, 0.39, 0.35$ that best fits the data.\label{fexp}}
\end{figure}

\section{Experiments}\label{Sexperiment}

In this section we discuss two experiments we performed with the prototype implementation of our method for approximate nearest neighbor search described in section~\ref{Sann}, to validate the theoretical results of our contributions.
In the first experiment, we computed the average value of the $k$ nearest neigbors needed to check in the projected space in order to get an actual nearest neighbor in the original space in a worst-case dataset for the ANN problem, and we confirmed that it is indeed sublinear in $n$. In the second experiment, we made an ANN query time and memory usage comparison against E2LSH using both artificial and real life datasets.

\subsection{Validation of $k$}
In this section we present an experimental verification of our approach. 
We show that the number $k$ of nearest neighbors in the projection space that we need to examine in order to find
an approximate nearest neighbor in the original space depends
sublinearly on $n$, thus validating in practice lemma~\ref{lemBBD}.

\paragraph*{Datasets} 
We generated our own synthetic datasets and query points. We decided 
to follow two different procedures for data generation.
First of all, as in \cite{DI04}, we followed the ``planted nearest neighbor model''. 
This model guarantees, for each query point $q$, the existence
of a few approximate nearest 
neighbors while keeping all others points sufficiently far from $q$.
The benefit of this approach 
is that it represents a typical ANN search scenario, where
for each point there exist only a handful 
approximate nearest neighbors. In contrast, in a uniformly generated
dataset, all points tend to be equidistant to each other
in high dimensions, which is quite unrealistic.

In order to generate the dataset, first we create a set $Q$
of query points chosen uniformly at 
random in $\mathbb{R}^d$. Then, for each point $q\in Q$,
we generate a single point $p$ at distance 
$R$ from $q$, which will be its single (approximate) nearest neighbor.
Then, we create more points at distance
$\ge (1+\epsilon)R$ from $q$, while making sure that they shall not be
closer than $(1+\epsilon)R$ to any other query point $q' \in Q, q' \neq q$.
This dataset now has the property that every query point has exactly
one approximate nearest neighbor, 
while all other points are at distance $\ge (1+\epsilon)R$.

We fix $R=2$, let $\epsilon \in \{0.1, 0.2, 0.5\},
d = \{200, 500\}$ and the total number of points
$n \in \{10^4, 2\times10^4, \ldots, 5\times10^4, 5.5\times10^4, 6\times10^4, 6.5\times10^4, \ldots, 10^5\}$. 
For each combination of the above we created a dataset $X$ from a set $Q$ of $100$ query points where each query 
coordinate was chosen uniformly at random in the range $[-20, 20]$.

The second type of datasets consisted again of sets of $100$ query points in $\mathbb{R}^d$ where 
each coordinate was chosen uniformly at random in the range $[-20, 20]$. Each query point was paired with a 
random variable $\sigma^2_q$ uniformly distributed in $[15, 25]$ and together they specified a gaussian distribution 
in $\mathbb{R}^d$ of mean value $\mu = q$ and variance $\sigma^2_q$ per coordinate. For each distribution we drew $n$ 
points in the same set as was previously specified.

\begin{figure}[!t] 
\centerline{\includegraphics[scale=0.42]{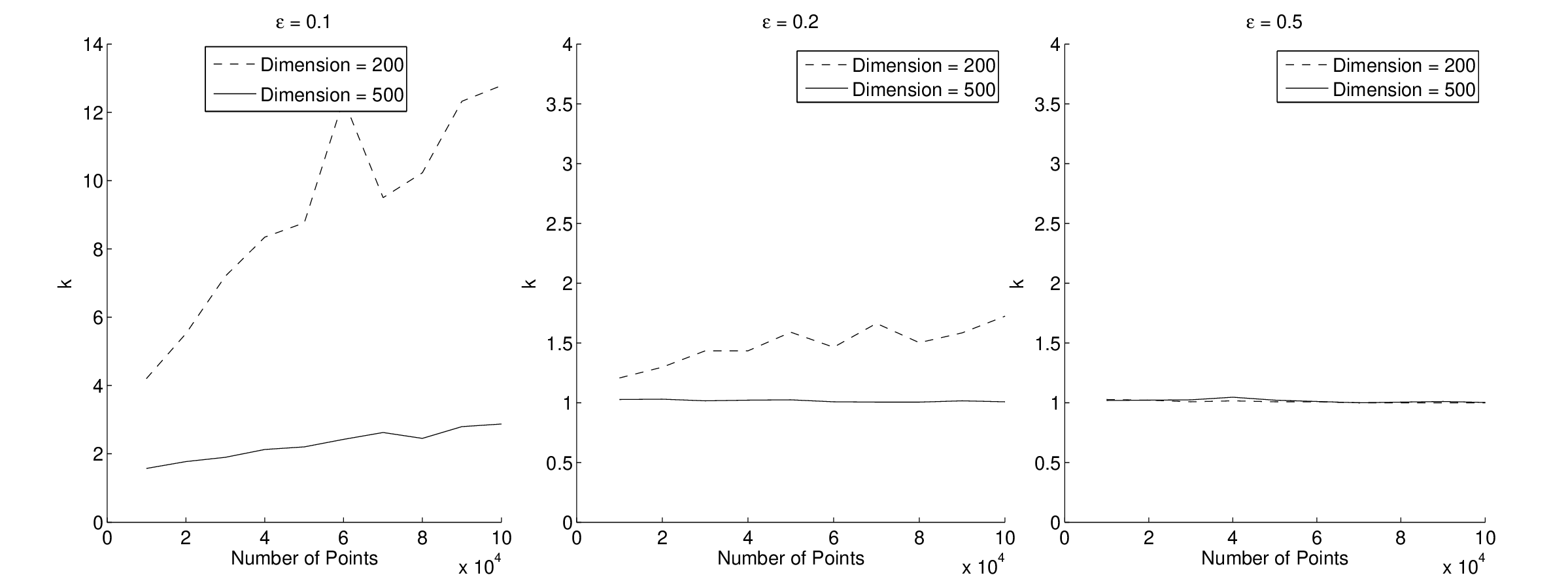}}
\caption{Plot of k as $n$ increases for the gaussian datasets. 
We see how increasing the number of approximate nearest neighbors in this case decreases
the value of $k$\label{fexpgm}}
\end{figure}
\paragraph*{Scenario}
We performed the following experiment for the ``planted nearest neighbor model''.
In each dataset $X$, we consider, for every query point $q$,
its unique (approximate) nearest neighbor $p \in X$. 
Then we use a random mapping $f$ from $\RR^d$ to a Euclidean space
of lower dimension
$d' = \frac{\log n}{\log\log n}$ using a gaussian matrix $G$, where each 
entry $G_{ij} \sim N(0,1)$. This matrix guarantees a
low distortion embedding \cite{IM98}. Then, we perform a range query centered 
at $f(q)$ with radius $\|f(q) - f(p)\|$ in $f(X)$:
we denote by $rank_q(p)$ the number of points found.
Then, exactly $rank_q(p)$ points are needed to be selected
in the worst case as $k$-nearest neighbors of $f(q)$ in order for the
approximate nearest neighbor $f(p)$ to be among them, so $k=rank_q(p)$.

For the datasets with the gaussian distributions we compute again the maximum number 
of points $k$ needed to visit in the lower-dimensional space in order 
to find an $\epsilon$-approximate nearest neighbor of each query point $q$ in the 
original space. In this case the experiment works as follows: we find all the $\epsilon$-approximate 
nearest neighbors of a query point $q$. Let $S_q$ be the set containing for each query $q$ its $\epsilon$-$k$ANNs. 
Next, let $p_q = \operatorname*{arg\,min}_{p \in S} \|f(p) - f(q)\|$. Now as before we perform a range query 
centered at $f(q)$ with radius $\|f(q) - f(p_q)\|$. We consider as $k$ the number of points returned by this query.

\paragraph*{Results}
The ``planted nearest neighbor model'' datasets constitute a 
worst-case input for our approach since every query point has only one 
approximate nearest neighbor and has many points lying near the boundary of 
$(1+\epsilon)$. We expect that the number of $k$ approximate nearest neighbors needed 
to consider in this case will be higher than in the case of the gaussian distributions, 
but still expect the number to be considerably sublinear.

In Figure~\ref{fexp} we present the average value of $k$ as we increase the number of points $n$ 
for the planted nearest neighbor model. We can see that $k$ is indeed significantly smaller
than $n$. The line corresponding to the averages may not be smooth, which is unavoidable due to the 
random nature of the embedding, but it does have an intrinsic concavity, which shows that 
the dependency of $k$ on $n$ is sublinear. For comparison we also display the function $\sqrt{n}/2$, as well 
as a function of the form $n^\rho, \rho<1$ which was computed by SAGE that best fits the data per plot. 
The fitting was performed on the points in the range $[50000, 100000]$ as to better capture the asymptotic behavior. 
In Figure~\ref{fexpgm} we show again the average value of $k$ as we increase the number of points $n$ for 
the gaussian distribution datasets. As expected we see that the expected value of $k$ is much smaller than 
$n$ and also smaller than the expected value of $k$ in the worst-case scenario, which is the planted nearest neighbor model.

\subsection{ANN experiments}
In this section we present a preliminary comparison between our algorithm and the E2LSH \cite{AI05} 
implementation of the LSH framework for approximate nearest neighbor queries.

\paragraph*{Experiment Description}
We projected all the ``planted nearest neighbor'' datasets, down to $\frac{\log n}{\log \log n}$ dimensions. 
We remind the reader that these datasets were created to have a single approximate nearest neighbor for each query 
at distance $R$ and all other points at distance $> (1+\epsilon) R$. We then built a BBD-tree data structure on the 
projected space using the {\tt ANN} library \cite{Mou10} with the default settings.  Next, we measured the average time needed for 
each query $q$ to find its $\epsilon$-$k$ANNs, for $k=\sqrt{n}$, using the BBD-Tree data structure and then to select the first 
point at distance $\leq R$ out of the $k$ in the original space. 
We compare these times to the average times reported by {E2LSH} range queries for $R = 2$, when used from its default 
script for probability of success $0.95$. The script first performs an estimation of the best parameters for the dataset and 
then builds its data structure using these parameters. 
We required from the two approaches to have accuracy $>0.90$, which in our case means that in at least $90$ out of the $100$ 
queries they would manage to find the approximate nearest neighbor.
We also measured the maximum resident set size of each approach which translates to the maximum portion of the main memory (RAM) 
occupied by a process during its lifetime. This roughly corresponds to the size of the dataset plus the size of the data structure 
for the E2LSH implementation and to the size of the dataset plus the size of the embedded dataset plus the size of the data structure 
for our approach. 
\paragraph*{ANN Results}
It is clear from Figure~\ref{fann} that E2LSH is faster than our approach by a factor of $3$. However in Figure~\ref{fmem}, where 
we present the memory usage comparison between the two approaches, it is obvious that E2LSH also requires more space. Figure~\ref{fmem} 
also validates the linear space dependency of our embedding method.
A few points can be raised here. First of all, we supplied the appropriate range to the LSH implementation, 
which gave it an advantage, because typically that would have to be computed empirically. To counter that, 
we allowed our algorithm to stop its search in the original space when it encountered a 
point that was at distance $\leq R$ from the query point. 
Our approach was simpler and the bottleneck was in the computation of the closest point out 
of the $k$ returned from the BBD-Tree. We conjecture that we can choose better values for our parameters $d'$ and $k$. Lastly, 
the theoretical guarantees for the query time of LSH are better than ours, but we did perform better in terms of space usage as expected. 

\begin{figure}[!ht] 
\centerline{\includegraphics[scale=0.42]{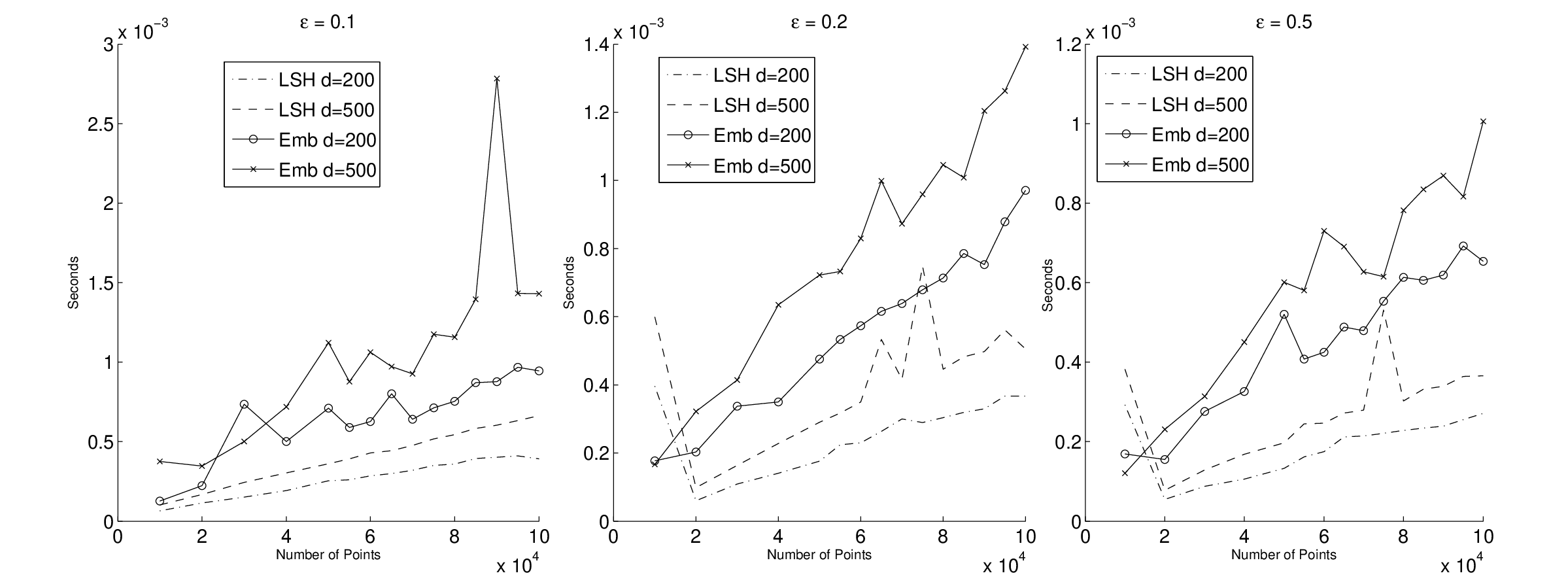}}
\caption{Comparison of average query time of our embedding approach against the E2LSH implementation.\label{fann}}
\end{figure}

\begin{figure}[!ht] 
\centerline{\includegraphics[scale=0.37]{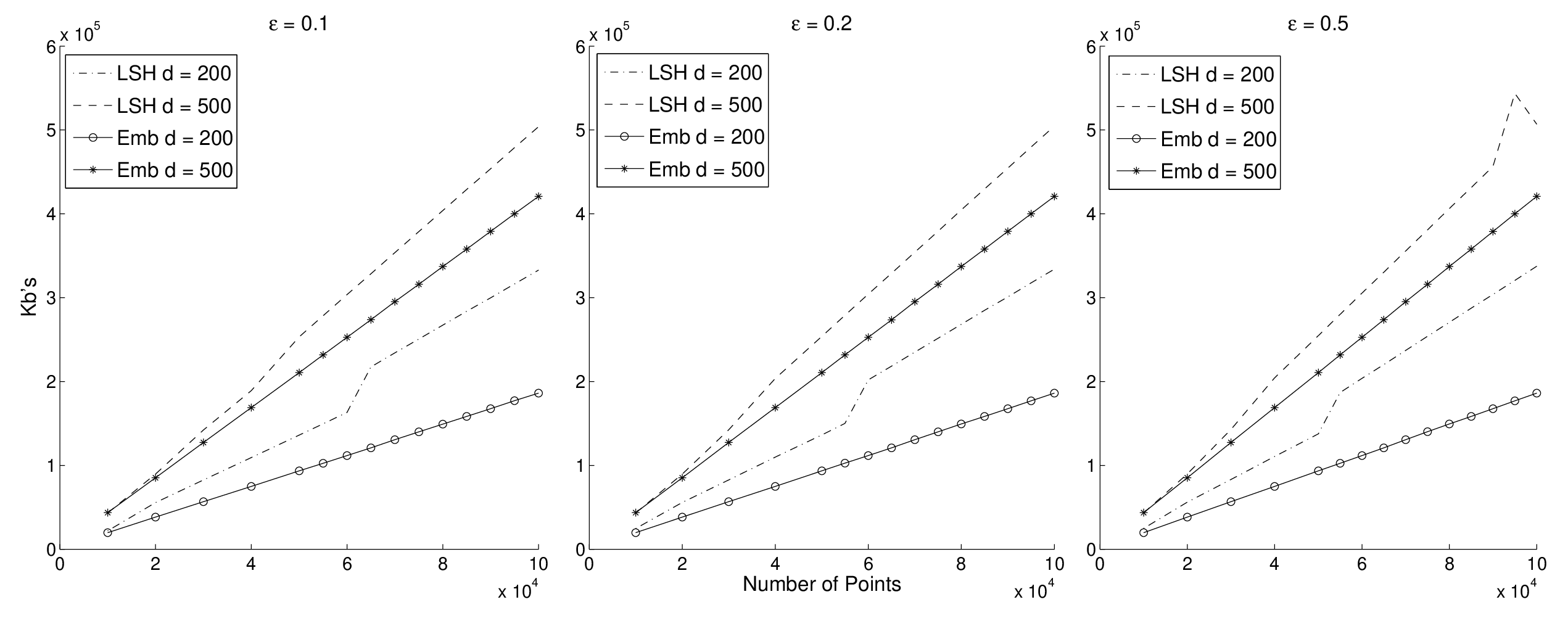}}
\caption{Comparison of memory usage of our embedding approach against the E2LSH implementation.\label{fmem}}
\end{figure}
\paragraph*{Real life dataset}
We also compared the two approaches using the ANN\_SIFT1M~\cite{sift} dataset which contains a collection of $1,000,000$ vectors in $128$ dimensions. 
This dataset also provides a query file containing $10,000$ vectors and a groundtruth file, which contains for each query the IDs of its $100$
nearest neighbors. These files allowed us to estimate the accuracy for each approach, as the fraction 
${\#hits}/{10000}$ where $\#hits$ denotes, for some query, the number of times one of its $100$ nearest neighbors were returned. The parameters 
of the two implementations were chosen empirically in order to achieve an accuracy of about $85\%$. For our approach we set the projection 
dimension $d' = 25$ and for the BBD-trees we specified $100$ points per leaf and $\epsilon = 0.5$ for the $\epsilon$-$k$ANNs queries. 
We also used $k = \sqrt{n}$. For the E2LSH implementation we specified the radius $R = 240$, $k = 18$ and $L = 250$. 
As before, we measured the average query time and the maximum resident set size. Our approach required an average of $171.59$msec per query, 
whilst E2LSH required $51.96$msec. However our memory footprint was about $1.256$ Gbytes and E2LSH used about $4.781$ Gbytes.

\section{Open questions}\label{Sopen}

The present work has emphasized asymptotic complexity bounds, and showed that rather simple methods, carefully combined with a new embedding approach, can achieve almost record query times with optimal space usage.
However, it should still be possible to enhance the practical performance of our method so as to unleash the potential of our approach and fully exploit its simplicity.
This is the topic of future work, along with a detailed comparative study with other optimized implementations, which is beyond the scope of this paper.

In particular, checking the real distance of the query point to the neighbors, while 
performing an $\epsilon$-$k$ANNs search in
the projection space, is more efficient in practice than naively scanning the returned sequence of 
$k$-approximate nearest neighbors, and looking for the closest point in the initial space. Moreover, our algorithm does not exploit the fact that 
BBD-trees return a sequence and not simply a set of neighbors.  

Our embedding approach probably has further applications. One possible application is in computing the $k$-th approximate nearest neighbor.
The problem may reduce to 
computing all neighbors between the $i$-th and the $j$-th nearest neighbors in a space of significantly smaller 
dimension for some appropriate values $i<k<j$. Other possible applications include computing the approximate minimum spanning tree, or 
the closest pair of points.

\bibliographystyle{alpha}
\bibliography{jlann}

\end{document}